\ifx\SODAVer\undefined
   \newcommand{\InSODAVer}[1]{}%
   \newcommand{\InFullVer}[1]{#1}%
\else
   
   \newcommand{\InSODAVer}[1]{#1}%
   \newcommand{\InFullVer}[1]{}%
\fi

\documentclass[11pt]{article}%
\usepackage[in]{fullpage}%

\usepackage{graphicx}%
\usepackage{amsmath}%
\usepackage{color}%
\usepackage{xspace}%
\usepackage{paralist}
\usepackage{mleftright}%
\usepackage{amssymb}%
\usepackage{xcolor}%
\usepackage{caption}%
\usepackage{stmaryrd}%

 \usepackage[amsmath,thmmarks]{ntheorem}%
 \theoremseparator{.}%

\usepackage{hyperref}%
\hypersetup{%
   breaklinks,%
   colorlinks=true,%
   linkcolor=[rgb]{0.45,0.0,0.0},%
   citecolor=[rgb]{0,0,0.45}
}

\usepackage[english]{babel}
\usepackage[utf8]{inputenc}
\usepackage[algo2e,boxed,linesnumbered,noend]{algorithm2e}
\usepackage[noend]{algpseudocode}
\usepackage{subcaption}

\numberwithin{figure}{section}%
\numberwithin{table}{section}%
\numberwithin{equation}{section}%

  \theoremstyle{plain}%
  \newtheorem{theorem}{Theorem}[section]
  
  \newtheorem{lemma}[theorem]{Lemma}%

  \newtheorem{corollary}[theorem]{Corollary}
   \newtheorem{problem}{Problem}[section]%
   \newtheorem{remark}[theorem]{Remark}

\newcommand{\HLinkShort}[2]{\hyperref[#2]{#1\ref*{#2}}}
\newcommand{\HLink}[2]{\hyperref[#2]{#1~\ref*{#2}}}
\newcommand{\HLinkPage}[2]{\hyperref[#2]{#1~\ref*{#2}%
      $_\text{p\pageref{#2}}$}}
\newcommand{\HLinkPageOnly}[1]{\hyperref[#1]{Page~\refpage*{#1}%
      $_\text{p\pageref{#1}}$}}

\newcommand{\HLinkSuffix}[3]{\hyperref[#2]{#1\ref*{#2}{#3}}}
\newcommand{\HLinkPageSuffix}[3]{\hyperref[#2]{#1\ref*{#2}%
      #3$_\text{p\pageref{#2}}$}}

\newcommand{\algolab}[1]{\label{algo:#1}}
\newcommand{\algoref}[1]{\HLink{Algorithm}{algo:#1}}%

\newcommand{\figlab}[1]{\label{fig:#1}}
\newcommand{\figref}[1]{\HLink{Figure}{fig:#1}}

\providecommand{\lemlab}[1]{\label{xlemma:#1}}
\renewcommand{\lemlab}[1]{\label{xlemma:#1}}
\newcommand{\lemref}[1]{\HLink{Lemma}{xlemma:#1}}%

\newcommand{\corlab}[1]{\label{corollary:#1}}
\newcommand{\corref}[1]{\HLink{Corollary}{corollary:#1}}%

\newcommand{\thmlab}[1]{{\label{theo:#1}}}
\newcommand{\thmref}[1]{\HLink{Theorem}{theo:#1}}

\newcommand{\seclab}[1]{\label{sec:#1}}
\newcommand{\secref}[1]{\HLink{Section}{sec:#1}}

\providecommand{\eqlab}[1]{}%
\renewcommand{\eqlab}[1]{\label{equation:#1}}

\newcommand{\myqedsymbol}{\rule{2mm}{2mm}}

\theoremheaderfont{\em}%
\theorembodyfont{\upshape}%
\theoremstyle{nonumberplain} \theoremseparator{}
\theoremsymbol{\myqedsymbol} \newtheorem{proof}{Proof:}
  
\newcommand{\etal}{\textit{et~al.}\xspace}

\DefineNamedColor{named}{RedViolet} {cmyk}{0.07,0.90,0,0.34}

\newcommand{\eps}{{\varepsilon}}%

\newcommand{\atgen}{\symbol{'100}}
\newcommand{\BenThanks}[1]{\thanks{Department of Computer Science;
      University of Texas at Dallas; 
      {\tt benjamin.raichel\atgen{}utdallas.edu}; {\tt
         \url{http://utdallas.edu/\string~benjamin.raichel}.} #1}}
\newcommand{\ChenglinThanks}[1]{\thanks{Department of Computer Science;
      University of Texas at Dallas; 
      {\tt cxf160130\atgen{}utdallas.edu}. #1}}
\newcommand{\GregoryThanks}[1]{\thanks{Department of Computer Science;
      University of Texas at Dallas; 
      {\tt greg.vanbuskirk\atgen{}utdallas.edu}. #1}}

\newcommand{\mvd}{\textsf{MVD}\xspace}
\newcommand{\mvdd}{\textsf{MVDD}\xspace}
\newcommand{\mvid}{\textsf{MVID}\xspace}
\newcommand{\gmd}{\textsf{GMVD}\xspace}
\newcommand{\gmvd}{\textsf{GMVD}\xspace}

\newcommand{\gmvdd}{\textsf{GMVDD}\xspace}
\newcommand{\gmvid}{\textsf{GMVID}\xspace}
\newcommand{\mcut}{\textsf{MULTICUT}\xspace}
\newcommand{\lcut}{\textsf{LB-CUT}\xspace}

\newcommand{\csp}{\mathsf{csp}}

\newcommand{\remove}[1]{}
\begin{document}

\title{Metric Violation Distance: Revisited and Extended}

\author{%
Chenglin Fan\ChenglinThanks{Work on this paper was partially
      supported by NSF CRII Award 1566137 and CAREER Award 1750780.}
\and
Benjamin Raichel\BenThanks{Work on this paper was partially
      supported by NSF CRII Award 1566137 and CAREER Award 1750780.}
      \and
Gregory Van Buskirk\GregoryThanks{Work on this paper was partially
      supported by NSF CRII Award 1566137 and CAREER Award 1750780.}   
}

\date{}

\maketitle


\begin{abstract}
Metric data plays an important role in various settings such as metric-based indexing, clustering, classification, and approximation algorithms in general. Due to measurement error, noise, or an inability to completely gather all the data, a collection of distances may not satisfy the basic metric requirements, most notably the triangle inequality. Thus last year the authors introduced the Metric Violation Distance (\mvd) problem \cite{frv-mvdha-18}, where the input is an undirected and positively-weighted complete graph, and the goal is to identify a minimum cardinality subset of edges whose weights can be modified such that the resulting graph is its own metric completion. This problem was shown to be APX-hard, and moreover an $O(OPT^{1/3})$-approximation was shown, where $OPT$ is the size of the optimal solution.

In this paper we introduce the \emph{Generalized Metric Violation Distance} (\gmvd) problem, where the goal is the same, but the input graph is no longer required to be complete. For \gmvd we prove stronger hardness results, and provide a significantly faster approximation algorithm with an improved approximation guarantee. In particular, we give an approximation-preserving reduction from the well studied \mcut problem, which is hard to approximate within any constant factor assuming the Unique Games Conjecture. Our approximation factor depends on \emph{deficit} values, which for a given cycle is the largest single edge weight minus the sum of the weights of all its other edges. Note that no cycle has positive deficit in a metric complete graph. We give an $O(c \log n)$-approximation algorithm for \gmvd, where $c$ is the number of distinct positive cycle deficit values in the input graph.
\end{abstract}

\thispagestyle{empty}
\newpage
\setcounter{page}{1}


\section{Introduction}

Given a large collection of data points for which there is some underlying notion of ``distance'' between pairs of points, 
one naturally may wish to perform any number of computational tasks over the data by using these distances (e.g.\ clustering).  
The ability to perform these tasks highly depends on the structure these distances obey. 
There are many settings where the underlying distances arise from a metric space or are at least well modeled by one. 
Such cases are fortuitous, as certain tasks become provably easier over metric data (e.g.\ approximating the optimal TSP tour), 
and moreover they allow us to use a number of computational tools such as metric embeddings. 

In this paper we consider the \emph{General Metric Violation Distance} (\gmvd) problem, 
where the input is an undirected and positively-weighted graph, and the goal is to identify a minimum cardinality subset of edges whose weights can be modified such that 
in the resulting graph all edge weights are the same as those in the metric completion of the graph, i.e.\ each edge is its own shortest path. 
Thus viewing the vertices as data points and edge weights as distances, \gmvd models the problem of finding the nearest metric, where nearest is defined by the minimum number of distances values that need to be altered.
For non-metric data sets, when this nearest metric is close, finding it allows one to take advantage of all the benefits of metric spaces.
As an alternative motivation, consider the case where the underlying distance function is truly metric, though the process by which distances where obtained occasionally produced measurement errors.  
In this case an algorithm for \gmvd serves as an approach to recover the true distance function.

Last year the authors introduced the Metric Violation Distance (\mvd) problem \cite{frv-mvdha-18}, which is a restriction of \gmvd to instances where the input graph is complete. 
Thus \mvd models the situation where one is given complete (though possibly erroneous) distance information. 
However, such complete information may not always be available (e.g.\ when measuring distances is expensive).
By not requiring a complete graph, \gmvd on the other hand can effectively model such incomplete data situations.
Moreover, by considering the more general \gmvd problem, in this paper we are able to give stronger hardness results, 
and to provide a faster approximation algorithm with an improved approximation factor, when compared to what was shown in \cite{frv-mvdha-18} for \mvd.

\paragraph{Related Work.}
Naturally, the most relevant previous work is \cite{frv-mvdha-18}, where the authors introduced the \mvd problem, 
which as described above is a special case of the \gmvd problem considered in this paper. 
For \mvd the authors gave an approximation preserving reduction from Vertex Cover, hence showing \mvd is APX-hard, 
and moreover is hard to approximate within a factor of $2$ assuming the Unique Games Conjecture (UGC) \cite{Khot02a}. 
The authors then provided an $O(OPT^{1/3})$-approximation, where $OPT$ is the size of the optimal solution.
The running time of this algorithm is $\Omega(n^6)$, as it enumerates all cycles of length $\leq 6$, 
and so the algorithm may not be practical from an implementation standpoint. 
The same hardness and approximation results were also shown for \mvid, 
which is the variant where edge weights are only allowed to increase.
The problem is polynomial time solvable when weights are only allowed to decrease.

\gmvd is also related to a large number of other previously studied problems.  
A short list includes:
metric nearness, seeking the metric minimizing the sum of distance value changes \cite{BrickellDST08}; 
metric embedding with outliers, seeking the fewest points whose removal creates a metric \cite{sww-meo-17}; 
matrix completion, seeking to fill missing matrix entries to produce a low rank \cite{cr-emcco-12}; and many more.  
We refer the reader to \cite{frv-mvdha-18} for a more detailed discussion of these and other problems.  

Instead here we focus on the deep connections to certain cutting problems, which underly several results in this paper, 
and which were not discussed in \cite{frv-mvdha-18}. 
In particular, our problem is closely related to \mcut 
where, given a weighted graph $G$ with $k$ marked $(s_i,t_i)$ vertex pairs, the goal is to find a minimum
weight subset of edges $S$ such that there is no path in $G\setminus S$ between $s_i$ and $t_i$ for any $1\leq i\leq k$. 
\mcut has been extensively studied, both for directed and undirected graphs.
For undirected graphs, the problem captures vertex cover even when $G$ is a tree and hence is APX-Hard. Moreover, assuming UGC there is no constant factor approximation \cite{ChawlaKKRS06}.
In general, the best known approximation factor is $O(\log k)$~\cite{GargVY96}, 
which improves to an $O(r)$-approximation when $G$ excludes $K_r$ as a minor \cite{AbrahamGGNT14}.
%
For directed graphs the problem is more challenging. It is hard to approximate within $2^{\Omega(\log ^{1-\epsilon} n)}$~\cite{ChuzhoyK09} assuming NP$\neq$ZPP, 
and the best approximation factor known is $\tilde{O}(n^{11/23})$~\cite{AgarwalAC07}, where $n$ is number of vertices in $G$. 
Significantly improved approximations are known when the demand graph (determined by viewing the $(s_i,t_i)$ pairs as edges) excludes certain induced subgraphs~\cite{ChekuriM17}.
%
%

Another closely related problem is Length Bounded Cut (\lcut), where given a value $L$ and an $(s,t)$ pair in a graph $G$,  
the goal is to delete the minimum number of edges such that there is no path between $s$ and $t$ with length $\leq L$. 
For \lcut the best known approximation factor is $O (\min \{L,n^2/L^2,\sqrt{m}\})$~\cite{BaierEHKKPSS10}, when $G$ has $n$ vertices and $m$ edges, 
where the factor $L$-approximation follows by simply repeatedly computing and removing minimum length $s$ to $t$ paths until $d(s,t)>L$.
For any fixed $L$, Lee~\cite{l-ihcifp-17} showed that it is hard to approximate within a factor of $\Omega(\sqrt{L})$ in undirected graphs, 
and within a factor of $\Omega(L)$ for directed graphs. (Kuhnle \etal~\cite{KuhnleCT18} showed a similar result for the directed case of Length Bounded Multicut assuming NP$\neq$BPP. 
Note the $L$-approximation of \cite{BaierEHKKPSS10} also applies to Length Bounded Multicut.)

\paragraph{Our Results.}
In this paper we introduce the \gmvd problem. 
This a generalization of the \mvd problem introduced last year by the authors, to allow missing edges, 
thus modeling the scenario of incomplete data.  
In addition to defining the \gmvd problem, we prove the following results.
\begin{compactitem}
 \item In \secref{cover} we prove a key structural result for the \gmvd problem. 
       Define an unbalanced cycle to be a cycle whose largest edge weight 
       is larger than the sum of the weights of all its other edges. 
       Then we argue that in order for a subset of edges to be a valid solution to \gmvd, 
       it suffices for the subset to cover all unbalanced cycles. 
       As this condition is also trivially necessary, 
       it gives a complete characterization of the \gmvd solution set. 
 \item In \secref{hard} we give polynomial-time approximation-preserving reductions from \mcut and \lcut to \gmvd.  
       This connection to the well studied \mcut problem is interesting in its own right, though it also implies 
       that \gmvd is NP-hard, and moreover cannot be approximated within any constant factor assuming UGC.
       Also recall that in general the best known approximation factor for \mcut is $O(\log n)$.
       Our reduction from \lcut implies that, for any fixed $L$, the set of instances of \gmvd with maximum edge weight $L$
       (and minimum weight 1) are hard to approximate within a factor of $\Omega(\sqrt{L})$.
 \item In \secref{approx} we give an approximation algorithm for \gmvd.  
       Our approximation factor depends on \emph{deficit} values, which for a given cycle is the largest single edge weight minus the sum of the weights of all its other edges. 
       Note that the unbalanced cycles are precisely those with positive deficit.
       We give an $O(c \log n)$-approximation algorithm for \gmvd, where $c$ is the number of distinct positive cycle deficit values in the input graph.
       Note there are several natural cases when $c$ is small: when the number of unbalanced cycles is small; 
       or for integer edge weights when all unbalanced cycles are within a bounded amount from being balanced; or alternatively when the maximum edge weight is bounded.
       Note also that our reductions in \secref{hard} suggest that the terms in our approximation factor may each be necessary in some form.  
       In particular, for the $\log n$ term recall that we reduce from \mcut, whose best known approximation factor is $O(\log n)$.   
       Moreover, for integer weights, if the maximum edge weight is $L$ then $c\leq L$, 
       and so our reduction from \lcut implies such instances are hard to approximate within a factor of $\Omega(\sqrt{c})$.
       
       Our algorithm runs in $O((n^3+m^2)\cdot OPT\cdot c \log n)$ time, and is practical in the sense that it only relies on basic counting and shortest path computations. 
       It should be noted that our new algorithm also work for \mvd, as it is a special case of \gmvd. 
       Conversely, it is not obvious that the $OPT^{1/3}$-approximation in~\cite{frv-mvdha-18} for \mvd would apply to \gmvd, 
       and moreover from a running time perspective, that algorithm must enumerate all cycles of length $\leq 6$.
\end{compactitem}
As was done in~\cite{frv-mvdha-18} for \mvd, throughout the paper we also consider the variant of \gmvd where edge weights are only allowed to increase, denoted \gmvid.
For this problem show the same hardness and approximation results described above for \gmvd. (We also get a similar sufficient condition in \secref{cover}, 
except that cycles must be covered with edges other than their largest weight edge.)
In fact, we argue in \secref{hard} that there is a polynomial time reduction from \gmvid to \gmvd. 
It should be noted, however, that it is not clear such a reduction from \mvid to \mvd holds.


\section{Preliminaries}\seclab{prelims}

Define a \emph{dissimilarity graph} as any complete, undirected, and positively-weighted graph, $G=(V,E)$. 
Throughout we use $w(e)$ to denote the input weight of any edge $e\in E$, and for any subset $F\subseteq E$ we let $w(F)$ denote $\sum_{e\in F} w(e)$.
We write $n=|V|$ and $m=|E|$, and for simplicity we assume $G$ is connected and so we always have $n-1\leq m$.
We say a dissimilarity graph $G$ is a \emph{metric graph} if it is its own metric completion, 
i.e.\ the weight $w(e)$ on each edge $e\in E$ is the shortest path distance between its endpoints.
The following problem was studied in \cite{frv-mvdha-18}.

\begin{problem}\textsc{(Metric Violation Distance (\mvd))}
Given a dissimilarity graph $G$, compute a minimum size set $S$ of edges whose weights can be modified to convert $G$ into a metric graph.
\end{problem} 

As discussed in \cite{frv-mvdha-18}, the \mvd problem can be equivalently formulated in terms of distance matrices, though for simplicity here we stick to the graph formulation.\footnote{
In \cite{frv-mvdha-18} the \mvd problem defined here is actually referred to as \gmvd, where the $G$ signifies the `graph' variant.  Instead here we use \gmvd to refer to a more `general' problem.}
In \cite{frv-mvdha-18}, two restricted variants \mvdd and \mvid were also considered, where edge weights are only allowed to be decreased and increased, respectively.

In this paper we consider a more general version of the above problems, where the graph $G=(V,E)$ is undirected and positively weighted, though is not necessarily complete. 
Define a \emph{general dissimilarity graph} to be any undirected, and positively-weighted graph, and a \emph{general metric graph} to be any general dissimilarity graph, $G=(V,E)$,
such that the weight $w(e)$ on each edge $e\in E$ is the weight of $e$ in the metric completion of $G$, i.e.\ it is the shortest path distance between its endpoints.
(The word `general' is sometimes dropped if understood from the context.)


\begin{problem}\textsc{(General Metric Violation Distance (\gmvd))}
Given a general dissimilarity graph $G$, compute a minimum size set $S$ of edges whose weights can be modified to convert $G$ into a general metric graph.
\end{problem} 
Again, we define two restricted variants \gmvdd and \gmvid, where edge weights are only allowed to be decreased and increased, respectively.
In \cite{frv-mvdha-18} it was argued that the optimal solution to \mvdd is precisely the set of edges whose weight is larger than the shortest path distance between its endpoints. 
Hence \mvdd is polynomial time solvable by computing the all pairs shortest paths distances, and it is not hard to see this similarly holds for \gmvdd. 
Thus for the remainder of the paper we only consider the \gmvd and \gmvid problems.


\subsection*{Covering Cycles.}
Given an undirected graph $G=(V,E)$, a subgraph $C=(V',E')$ is called a $k$-cycle if $|V'|=|E'|=k$, and the subgraph is connected with every vertex having degree exactly $2$. 
We often overload this notation and use $C$ to denote either the cyclically ordered list of vertices or edges from this subgraph. 
We also often write $C\setminus e$ to denote the set of edges of $C$ after removing the edge $e$, and $\pi(C\setminus e)$ to denote the corresponding induced path between the endpoints of $e$.
Call any edge from $G$ connecting two vertices which are non-adjacent in a given cycle, a \emph{chord} of that cycle.
Given a $k$-cycle in $G$, if the weight of a single edge is strictly larger than the sum of the weights of the other edges in the cycle, we say it is an \emph{unbalanced $k$-cycle}. 
For a given unbalanced $k$-cycle, call its largest weight edge its \emph{top edge} and the other edges of the cycle the \emph{non-top} edges.
Let the \emph{deficit} of a cycle $C$, denoted $\delta(C)$, be equal to the weight of its top edge minus the sum of the weights of all other edges in $C$.  
Similarly, let $\delta(G)$ denote the maximum value of $\delta(C)$ over all cycles.
Note the set of unbalanced cycles is equivalently the set of cycles with strictly positive deficit.

We define three notions of covering unbalanced cycles.  
Specifically, let $\mathbb{C}$ be any collection of unbalanced cycles from a positively-weighted graph $G=(V,E)$.
An edge subset $F\subseteq E$ is a 
\begin{inparaenum}[(i)]
\item \emph{regular cover}, 
\item \emph{non-top cover}, or
\item \emph{top cover}
\end{inparaenum}
of $\mathbb{C}$, if $F$ contains at least one 
\begin{inparaenum}[(i)]
\item edge,
\item non-top edge, or
\item top edge
\end{inparaenum}
of every unbalanced cycle in $\mathbb{C}$.  
In particular, if $\mathbb{C}$ is the set of \emph{all} unbalanced cycles in $G$ then we say $F$ 
\begin{inparaenum}[(i)]
(i) regular covers, 
(ii) non-top covers, or
(iii) top covers
\end{inparaenum}
$G$, respectively.

It is not hard to see that in order for a dissimilarity graph to be a metric graph, there cannot be any unbalanced cycles, 
and the same is true for general dissimilarity graphs being general metric graphs.
(See \cite{frv-mvdha-18} for a proof for dissimilarity graphs, which trivially extends to the general case.)
Thus a solution to \mvd or \gmvd must necessarily cover all unbalanced cycles (and similarly non-top cover for \mvid and \gmvid, and top cover for \mvdd and \gmvdd).
What is more surprising is that for \mvd this is also sufficient.

\begin{theorem}[\cite{frv-mvdha-18}]\thmlab{coverOldInc}
 If $G$ is an instance of \mvid and $S$ is a non-top cover of all unbalanced cycles, then $G$ can be converted into a metric graph by only increasing weights of edges in $S$.
\end{theorem}
\begin{theorem}[\cite{frv-mvdha-18}]\thmlab{coverOldGen}
 If $G$ is an instance of \mvd and $S$ is a regular cover of all unbalanced cycles, then $G$ can be converted into a metric graph by only changing weights of edges in $S$.
\end{theorem}
In \secref{cover} we generalize the above theorems to \gmvid and \gmvd.

\subsection*{Feasibility Checking}
Given a general dissimilarity graph $G$, let $S$ be any subset of edges whose weights can be modified to convert $G$ into a general metric graph. 
We now show that there is a linear program which determines how to set the weights of the edges in $S$, 
and thus in the remainder of the paper we only need to focus on finding the set $S$.  
A similar linear program was described \cite{frv-mvdha-18}, though here some slight modifications are needed as $G$ is no longer assumed to be complete.

Let the vertices in $G=(V,E)$ be labeled $[n] = \{1,\ldots,n\}$. 
Then for all pairs $i,j\in [n]$, we define a (symmetric) variable $\alpha_{ij}=\alpha_{ji}$. 
To enforce that we get a general metric graph, we require that all triangle inequalities $\alpha_{ik}\leq \alpha_{ij}+\alpha_{jk}$ are satisfied.
As only edges in $S$ are allowed to change weights, we require for any $(i,j)\in E\setminus S$ that $\alpha_{ij} = w((i,j))$. We thus have the following.

\begin{equation*}
\begin{aligned}
& & & \alpha_{ij} = \alpha_{ji} = w((i,j))\; & &~~~\forall (i,j)\in E\setminus S \\
& & & \alpha_{ij} = \alpha_{ji} \geq 0\; & &~~~\forall (i,j)\notin E\setminus S \\
& & & \alpha_{ik}\leq \alpha_{ij}+\alpha_{jk}\; & &~~~\forall \text{ distinct triples } \{i,j,k\} \\
\end{aligned}
\end{equation*}
\newcommand{\lpproof}{%
We now argue the LP is feasible if and only if weights of edges in $S$ can be modified to convert $G$ into a general metric graph.
First, suppose we have such a set $S$, and that the edge weights in $S$ have already been modified such that $G$ is now a general metric graph.
Then we argue setting $\alpha_{ij} = w((i,j))$ for all $(i,j)\in E$, and otherwise setting $\alpha_{ij}$ equal to the shortest path distance from $i$ to $j$ in $G$, gives a feasible LP solution.
Specifically, by definition of general metric graphs, $\alpha_{ij} = w((i,j))$ is the shortest path distance between $i$ and $j$ in $G$, which is the same way we defined $\alpha_{ij}$ for all $(i,j)\notin E$. 
Thus as the triangle inequality holds for shortest path distances, $\alpha_{ik} = d(i,k) \leq d(i,j) + d(j,k) =\alpha_{ij}+\alpha_{jk}$, where $d(x,y)$ is the shortest path distance from $x$ to $y$.

Now suppose that the LP has a feasible solution, and set the weight of each edge $(i,j)\in S$ equal to the $\alpha_{i,j}$ value from the LP.
Consider any edge $(i,j)\in E$, and let $i = l_1,l_2,\ldots,l_k = j$ be the in order sequence of vertices in a shortest path from $i$ to $j$ in $G$.
By definition, to argue that $G$ (with the weights for edges in $S$ set to their $\alpha_{ij}$ values) is a generalized metric graph, we need to show $w(l_1,l_k)\leq \sum_{x=1}^{k-1} w(l_x,l_{x+1})$. 
However, by repeated application of the triangle inequality, $w(l_1,l_k) = \alpha_{l_1,l_k} \leq \sum_{x=1}^{k-1} \alpha_{l_x,l_{x+1}} = \sum_{x=1}^{k-1} w(l_x,l_{x+1})$, 
thus proving the claim.
}
\InFullVer{\lpproof}

\InSODAVer{
It is straightforward to argue the above LP is feasible if and only if weights of edges in $S$ can be modified to convert $G$ into a general metric graph, 
and for completeness the proof is given in the appendix.}%
Later in the paper we consider \gmvid, where entries are only allowed to be increased.
Note that the above linear program can trivially be modified to handle this case.
Namely, replace the second constraint with $\alpha_{ij} = \alpha_{ji}\geq w((i,j)) ~~ \forall (i,j)\in S$ and $\alpha_{ij} = \alpha_{ji}\geq 0 ~~ \forall (i,j)\notin E$.


\section{Covering is Sufficient}\seclab{cover}

Here we generalize \thmref{coverOldInc} and \thmref{coverOldGen} to the case of general dissimilarity graphs. 
These key structural results are needed to design our approximation algorithms. 
Also note our reduction from \gmvid to \gmvd in the next section depends on both \thmref{gmvid_iff} and \thmref{gmvd_iff} below.

\begin{theorem}\thmlab{gmvid_iff}
If $G$ is an instance of \gmvid and $S$ is a non-top cover of all unbalanced cycles, then $G$ can be converted into a metric graph by only increasing weights of edges in $S$.
\end{theorem}
\begin{proof}
For now assume all edge weights are integers and let $L$ denote the largest edge weight.  
We describe a procedure which only modifies the weights of edges in $S$, producing a new instance $G'$, such that
\begin{inparaenum}[(i)]
 \item $S$ remains a non-top cover of $G'$, 
 \item the weight of at least one edge strictly increases and none decrease, and
 \item no edge weight is ever increased above $L$.
\end{inparaenum} 
For any instance $G$ (with at least one unbalanced cycle) and non-top cover $S$, if we prove such a procedure exists whenever not all the edges in $S$ have weight $L$, 
then this will imply the lemma.
Specifically, after applying the procedure at most $|S|\cdot L$ times, all edge weights will be equal to $L$, 
and hence no unbalanced cycle (and so no unsatisfied triangle inequality) can remain, 
since otherwise $S$ would cover that unbalanced cycle with an edge of weight $L$, 
which is at least the weight of that cycle's top edge (i.e. the cycle is not actually unbalanced).
Moreover, this procedure only increases weights of edges in $S$, as desired.
Note also that as the number of steps in this existential argument was irrelevant, so long as it was finite, 
this procedure will imply the claim for rational input weights as well.

We now prove the above described procedure exists.  Let $G$ and $S$ be as in the lemma statement, 
and fix any unbalanced cycle, $C$, which can be assumed to exist, as otherwise the lemma is trivially true. 
Let $t$ be the top edge of $C$, and let $C\setminus t$ denote the set of non-top edges of $C$.
Let $F=(C\setminus t)\cap S$ be the set of non-top edges in $C$ that are covered by $S$. Note that $F$ is non-empty. 

If there exists any edge $f\in F$ whose weight can be increased without creating any new unbalanced cycle which is not non-top covered, then we increase $f$ by one, and move onto the next iteration of our procedure.  
Thus we can assume we have reached an iteration where no such edge exist. Thus for every edge $f\in F$, there must be some cycle $\mathbb{C}_f$, which is not non-top covered, and such that if we increased $f$ at all then $\mathbb{C}_f$ would become unbalanced. Note this implies $w(f) = \sum_{e\in (\mathbb{C}_f\setminus f)} w(e)$, and if increased $f$ would be the top edge of this new unbalanced cycle.

Let $(t, e_1, e_2,\cdots, e_{k})$ be a cyclic ordering of the edges in $C$. For any $e_i\in F$,  
let $\pi(\mathbb{C}_{e_i} \setminus e_i)$ denote the sub-path of cycle $\mathbb{C}_{e_i}$ starting at the common vertex of $e_{i-1}$ and $e_{i}$, and ending at the common vertex of $e_i$ and $e_{i+1}$. 
For any $e_i\in (C\setminus t)$, if $e_i\in F$ then define $\mathbb{P}_i=\pi(\mathbb{C}_{e_i} \setminus e_i)$, and otherwise define $\mathbb{P}_i=e_i$.
Consider the closed walk $\sigma = \mathbb{P}_1 \circ\cdots\circ \mathbb{P}_{k} \circ t$.
As $\sigma$ is closed it must contain some cycle $D$ which includes $t$.
Since $D$ is a cycle from $\sigma$, $w(D\setminus t) = w(D) - w(t) \leq w(\mathbb{P}_1) + \cdots + w(\mathbb{P}_{k}) = w(e_1) + \cdots + w(e_{k}) < w(t)$, where the last inequality follows since $C$ is unbalanced.
Thus by definition $D$ is unbalanced with top edge $t$, and so by assumption must be non-top covered by $S$, which in turn implies some edge in $\sigma\setminus t$ lies in $S$.  
However, every edge in $\sigma\setminus t$ is either a non-top edge of $C$ which is not in $S$, or a non-top edge of $\mathbb{C}_f$ for some $f\in F$, and we assumed all such cycles are not non-top covered.
Thus in either case we get a contradiction.

Therefore, as long as there are unbalanced cycles, each step of the procedure strictly increases an edge from $S$ and no edges get decreased. Since any newly created unbalanced cycle is still covered by $S$, conditions (i), (ii), and (iii) are satisfied.
\end{proof}


\newcommand{\helper}{
The following helper lemma will be required for our proof of \thmref{gmvd_iff}.
This lemma was originally proven for the \mvd problem in \cite{frv-mvdha-18}. 
Its proof works verbatim for the \gmvd problem, and for completeness is included here.

\begin{lemma}[\cite{frv-mvdha-18}] \lemlab{split}
 Let $G=(V,E)$ be an instance of \gmd. If $S$ is a regular cover of all unbalanced cycles, then $S$ can be partitioned into two disjoint sets $S^+$ and $S^-$ such that each unbalanced cycle is either non-top covered by $S^+$ or top covered by $S^-$.
\end{lemma}
\begin{proof}
Let $S$ be a regular cover of all unbalanced cycles and let $S^+$ and $S^-$ initially be empty sets.  
We now define an interative procedure, which in each iteration removes one edge from $S$ and adds it to either $S^+$ or $S^-$.
We maintain the invariant that each unbalanced cycle is either top covered by $S^-\cup S$ or non-top covered by $S^{+}\cup S$.
Thus, after a finite number of iterations, $S$ will be empty, 
and $S^+$ and $S^-$ will be two disjoint sets such that each unbalanced cycle is either non-top covered by $S^+$ or top covered by $S^-$.

Suppose at step $i$, we pick an edge $b$ from set $S$. There are two cases.
Case 1: if each unbalanced cycle is either top covered by $S^{-}\cup (S\setminus b)$ or non-top covered by $(S^{+} \cup b)\cup (S\setminus b)$, then we add $b$ to $S^{+}$.
Case 2: if each unbalanced cycle is either top covered by $(S^{-}\cup b)\cup (S\setminus b)$ or non-top covered by $S^{+}\cup (S\setminus b)$, then we add $b$ to $S^{-}$.
We now argue by contradiction that these are the only possible cases.

If Case 1 does not hold, then there must be an unbalanced cycle $\mathbb{C}_1$ which is neither top covered by $S^{-}\cup (S\setminus b)$ nor non-top covered by $(S^{+} \cup b)\cup (S\setminus b)$. As $\mathbb{C}_1$ must be top covered by $S^-\cup S$ or non-top covered by $S^{+}\cup S$ (by induction), this implies that $b$ must top cover $\mathbb{C}_1$.
If Case 2 does not hold then there is an unbalanced cycle $\mathbb{C}_2$ which is neither top covered $(S^{-}\cup b)\cup (S\setminus b)$ nor non-top covered by $S^{+}\cup (S\setminus b)$, 
and similarly this implies $b$ must non-top cover $\mathbb{C}_2$.
Now consider the closed walk $\sigma = \pi(\mathbb{C}_1\setminus b) \circ \pi(\mathbb{C}_2\setminus b)$ (see \figref{both3}). Let $t$ be the top edge of $\mathbb{C}_2$.
There exists a cycle $C$ only containing edges from $\sigma$ whose top edge is $t$, and is unbalanced because 
$w(t) > w(\mathbb{C}_2\setminus t) = w(\mathbb{C}_2\setminus \{t,b\}) + w(b) > w(\mathbb{C}_2\setminus \{t,b\}) + w(\mathbb{C}_1\setminus b)\geq w(C)-w(t)$. 
Observe that $(\mathbb{C}_1\setminus b)\cap(S^+\cup S) = \emptyset$ and $(\mathbb{C}_2\setminus t)\cap(S^+\cup (S\setminus b)) = \emptyset$ which implies $(C\setminus t)\cap (S^+\cup S) = \emptyset$ because $b\notin C$.
Additionally, we have that $t\notin (S^-\cup S)$ as otherwise $\mathbb{C}_2$ would be top covered by $(S^{-}\cup b)\cup (S\setminus b)$.
As $C$ must be top covered by $S^-\cup S$ or non-top covered by $S^{+}\cup S$ (again by induction), 
this give a contradiction as we showed $t\notin (S^-\cup S)$ and $(C\setminus t)\cap (S^+\cup S) = \emptyset$.

\begin{figure}[t]
    \centering
    \includegraphics[height=.24\linewidth]{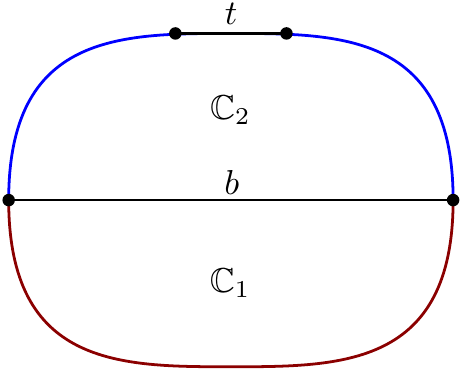}
    \caption{Edge $b$ top covers $\mathbb{C}_1$ and non-top covers $\mathbb{C}_2$.}
    \figlab{both3}
\end{figure}

Therefore, we can add $b$ to $S^+$ or $S^-$ (according to case 1 or 2) and remove it from $S$ 
such that each unbalanced cycle remains either top covered by $S^-\cup S$ or non-top covered by $S^{+}\cup S$.
Thus, after at most $|S_0|$ rounds, where $S_0$ is the initial state of $S$, $S^+$ and $S^-$ will be as in the lemma statement. 
\end{proof}
}

\InFullVer{\helper}
\InSODAVer{
The proof of the following theorem for \gmvd is similar to, albeit more intricate than, the above \gmvid proof.
To give our main results in the first 10 pages, its proof was moved to the appendix.
}

\begin{theorem}\thmlab{gmvd_iff}
If $G$ is an instance of \gmd and $S$ is a regular cover of all unbalanced cycles, then $G$ can be converted into a metric graph by only changing weights of edges in $S$.
\end{theorem}
\newcommand{\gmvdproof}{
\begin{proof}
For now assume all edge weights are integers and let $L$ denote the largest edge weight.
First use \lemref{split} to partition $S$ into two disjoint sets $S^+$ and $S^-$, such that every unbalanced cycle is either non-top covered by $S^+$ or top covered by $S^-$.
We describe a procedure producing a new instance $G'$, such that
\begin{inparaenum}[(i)]
 \item  every unbalanced cycle of $G'$ is either non-top covered by $S^+$ or top covered by $S^-$, 
 \item the weight of either one `$+$' edge strictly increases or one `$-$' edge strictly decreases (and no other edge weights are modified), and
 \item no edge weight is ever increased above $L$ or below $0$.
\end{inparaenum} 
Proving such a procedure exists for any instance $G$ and regular cover $S$, will imply the lemma.
Specifically, after applying the procedure at most $|S|\cdot L$ times, all edge weights  in $S^{+}$ will be equal to $L$, and all edge weights  in $S^{-}$ will be equal to $0$.
This implies there are no remaining unbalanced cycles, since otherwise the unbalanced cycle must either be non-top covered by an edge of weight $L$ or top covered by an edge weight $0$, 
in either case yielding the contradiction that the cycle was actually balanced.
Moreover, this procedure correctly modifies edge weights from $S$ according to their label.
As the number of steps in this existential argument was irrelevant, so long as it was finite, 
this procedure will imply the claim for rational input weights as well.

We now prove the above described procedure exists.  Let $G$ and $S$ be as in the lemma statement, 
and fix any unbalanced cycle, $C$, which can be assumed to exist as otherwise the lemma is trivially true. 
Let $t$ be the top edge of $C$, and let $C\setminus t$ denote the set of non-top edges of $C$. Let $F=(C\setminus t) \cap S^+$ be the set of non-top edges in $C$ that are covered by $S^+$. Note that $F$ is non-empty or $t\in S^-$.
If a cycle is not non-top covered by $S^{+}$ and not top covered by $S^{-}$, we say that cycle is \textit{uncovered}.

If there exists any edge $f\in F$ whose weight can be increased without creating any new unbalanced cycle which is uncovered, then we increase $f$ by one, and move onto the next iteration of our procedure.
Similarly, we decrease $t$ by one if $t\in S^-$ and decreasing $t$ does not create any new unbalanced cycle which is uncovered.
Thus we can assume we have reached an iteration where no such edge exist. Thus for every edge $f\in F$, there must be some cycle $\mathbb{C}_f$, which is uncovered, and such that if we increased $f$ at all then $\mathbb{C}_f$ would become unbalanced. Note this implies $w(f) = \sum_{e\in (\mathbb{C}_f\setminus f)} w(e)$, and if increased $f$ would be the top edge of this new unbalanced cycle. Similarly, if $t\in S^-$, then there must be some uncovered cycle $\mathbb{C}_t$ which would become unbalanced if $t$ were decreased any further.

We now argue the existence of an unbalanced cycle $C$, in an iteration where weights in $S^{+}$ cannot be increased nor weights in $S^{-}$ decreased as described above, 
will contradict the starting assumption that all unbalanced cycles are either non-top covered by $S^+$ or top covered by $S^-$.
We break the analysis into two cases. 
For both cases, let $(t, e_1, e_2,\cdots, e_{k})$ be a cyclic ordering of the edges in $C$.
For any $e_i\in F$, let $\pi(\mathbb{C}_{e_i} \setminus e_i)$ denote the sub-path of cycle $\mathbb{C}_{e_i}$ starting at the common vertex of $e_{i-1}$ and $e_{i}$, and ending at the common vertex of $e_i$ and $e_{i+1}$. 
For any $e_i\in (C\setminus t)$, if $e_i\in F$ then define $\mathbb{P}_i=\pi(\mathbb{C}_{e_i} \setminus e_i)$, and otherwise define $\mathbb{P}_i=e_i$.

\textbf{Case 1:} $t\notin S^-$.
Consider the closed walk $\sigma = \mathbb{P}_1 \circ\cdots\circ \mathbb{P}_{k} \circ t$.
As $\sigma$ is closed it must contain some cycle $D$ which includes $t$.
Since $D$ is a cycle from $\sigma$, $w(D\setminus t) = w(D) - w(t) \leq w(\mathbb{P}_1) + \cdots + w(\mathbb{P}_{k}) = w(e_1) + \cdots + w(e_{k}) < w(t)$, where the last inequality follows since $C$ is unbalanced.
Thus by definition $D$ is unbalanced with top edge $t$, and so by assumption must either be top covered by $S^-$ or non-top covered by $S^+$. Since $t\notin S^-$, this implies some edge in $\sigma\setminus t$ lies in $S^+$.  
However, every edge in $\sigma\setminus t$ is either a non-top edge of $C$ which is not in $S^+$, or a non-top edge of $\mathbb{C}_f$ for some $f\in F$, and we assumed all such cycles are not non-top covered.
Thus in either case we get a contradiction.

\textbf{Case 2:} $t\in S^-$.
For $1\leq i\leq k$, let $\mathbb{P}_i$ be as defined above, and
consider the cycle $\mathbb{C}_t$ which by definition is uncovered and would become unbalanced, with top edge $t'$, if $t$ were decreased.
Let $(t', g_1, g_2, \cdots, g_j, t, g_{j+1}, \cdots, g_l)$ be a cyclic ordering of the edges in $\mathbb{C}_t$.
Consider the closed walk $\sigma = t' \circ g_1 \circ\cdots\circ g_j \circ \mathbb{P}_1 \circ\cdots\circ \mathbb{P}_k \circ g_{j+1} \circ\cdots\circ g_l$.
As $\sigma$ is closed it must contain some cycle $D$ which includes $t'$.
Since $D$ is a cycle from $\sigma$,
$w(D\setminus t') = w(D) - w(t')
\leq w(g_1) + \cdots + w(g_{j}) + w(\mathbb{P}_1) + \cdots + w(\mathbb{P}_{k}) + w(g_{j+1}) + \cdots + w(g_{l})
= w(g_1) + \cdots + w(g_{j}) + w(e_1) + \cdots + w(e_{k}) + w(g_{j+1}) + \cdots + w(g_{l})
< w(g_1) + \cdots + w(g_{j}) + w(t) + w(g_{j+1}) + \cdots + w(g_{l})
= w(t')$, where the strict inequality follows since $C$ is unbalanced.
Thus by definition $D$ is unbalanced with top edge $t'$, and so by assumption must either be top covered by $S^-$ or non-top covered by $S^+$. 
This implies either an edge in some $\mathbb{P}_i$ lies in $S^+$, some edge $g_i$ lies in $S^+$, or $t'\in S^-$.
However, by assumption, each $\mathbb{P}_i$ is either an uncovered non-top edge from $C$ or is a path coming from an uncovered cycle $\mathbb{C}_i$, and hence no edge in $\mathbb{P}_i$ lies in $S^+$. 
Additionally, $\mathbb{C}_t$ is assumed to be uncovered and so no $g_i$ lies in $S^+$, and $t'$ is not in $S^-$.
Thus we get a contradiction.

Thus as long as unbalanced cycles remain, at least one edge weight from $S$ gets increased or decreased (and not above $L$ or below $0$) according to its label, 
and any newly created unbalanced cycle is already top covered by $S^{-}$ or non-top covered by $S^{+}$.
\end{proof}
}
\InFullVer{
\gmvdproof
}


\section{Hardness}\seclab{hard}

Previously, \cite{frv-mvdha-18} gave an approximation-preserving reduction from Vertex Cover to both \mvd and \mvid.  Thus both are APX-complete, and in particular are hard to approximate within a factor of $2-\eps$ for any $\eps>0$, assuming the Unique Games Conjecture (UGC) \cite{Khot02a}. In this section we give stronger hardness results for \gmvid and \gmvd by giving approximation-preserving reductions from \mcut and \lcut.

\begin{problem}[\mcut]
Given an undirected unweighted graph $G=(V,E)$ on $n=|V|$ vertices together with $k$ pairs of vertices $\{s_i , t_i \}^k_{i=1}$, compute a minimum size subset of edges $M\subseteq E$ whose removal disconnects all the demand pairs, i.e., in the subgraph $(V,E \setminus M)$ every $s_i$ is disconnected from its corresponding vertex $t_i$.
\end{problem} 
\cite{ChawlaKKRS06} proved that if UGC is true, then it is NP-hard to approximate \mcut within any constant factor $L>0$, and assuming a stronger version of the UGC, within $\Omega(\sqrt{\log\log n})$.
(Note that the version of \mcut in \cite{ChawlaKKRS06} allows weights, but as the weights are polynomial, the authors remark that their hardness proofs extend to the unweighted case.)


\begin{theorem}\thmlab{mcut_to_gmvid}
 There is an approximation-preserving, polynomial-time reduction from \mcut to \gmvid.
\end{theorem}
\begin{proof}
Let $G=(V,E)$ be an instance of \mcut with $k$ pairs of vertices $\{s_i , t_i \}^k_{i=1}$. First, if $(s_i,t_i)\in E$ for any $i$, then that edge must be included in the solution $M$.  
Thus we can assume no such edges exists in the \mcut instance, as assuming this can only make it harder to approximate the optimum value of the \mcut instance. 
We now construct an instance of \gmvid, $G'=(V',E')$. Let $V'=V$ and $E'=E\cup \{s_i , t_i \}^k_{i=1}$ where the edges in $E$ have weight $1$ and the edges $(s_i, t_i)$, for all $i\in [k]$, have weight $n=|V|$.

Observe that if a cycle in $G'$ has exactly one edge of weight $n$, then it must be unbalanced, since a cycle can have at most $n-1$ other edges (each of weight $1$).
The claim is that the converse is also true, that is any unbalanced cycle in $G'$ must have exactly one edge of weight $n$ (and note this is the top edge).  
Specifically, if the cycle has no $n$ weight edges, then all edges have weight $1$ and hence the cycle is balanced.
If the cycle has more than one edge with weight $n$, then the cycle is also balanced, since for any edge, the total weight of the other edges in the cycle is $\geq n$, and no edge has weight larger than $n$. 

Note that the edges from $G$ are exactly the weight one edges in $G'$, and thus the paths in $G$ are in one-to-one correspondence with the paths in $G'$ which consist of only weight one edges. 
Moreover, the weight $n$ edges in $G'$ are in one-to-correspondence with the $(s_i,t_i)$ pairs from $G$. 
Thus the cycles in $G'$ with exactly one weight $n$ edge followed paths of all weight one edges connecting their endpoints, which by the above are exactly the set of unbalanced cycles, are in one-to-one correspondence with paths between $(s_i,t_i)$ pairs from $G$.
Therefore, a minimum cardinality subset of edges which non-top cover all unbalanced cycles, i.e.\ an optimal solution to \gmvid, corresponds to a minimum cardinality subset of edges from $E$ which cover all paths from $s_i$ to $t_i$ for all $i$, i.e.\ an optimal solution to \mcut. 
%
\end{proof}

\begin{problem}[\lcut]
Given a value $L$ and an undirected unweighted graph $G=(V,E)$ with source $s$ and sink $t$, find a minimum size subset of edges $M\subseteq E$ such that no $s$-$t$-path of length less than or equal to $L$ remains in the graph after removing the edges in $M$.
\end{problem}

An instance of \lcut with length bound $L$, is referred to as an instance of the $L$-\lcut problem.
For any fixed $L$, Lee~\cite{l-ihcifp-17} showed that it is hard to approximate $L$-\lcut within a factor of $\Omega(\sqrt{L})$.  
\InSODAVer{The proof of the following theorem is similar to that of \thmref{mcut_to_gmvid} (actually it is simpler as there is only one $(s,t)$ pair), 
and as such has been moved to the appendix.}

\begin{theorem}\thmlab{lbcut_to_gmvid}
For any fixed value $L$, there is an approximation-preserving, polynomial-time reduction from $L$-\lcut to \gmvid.
\end{theorem}
\newcommand{\lbcutproof}{
\begin{proof}
Let $G=(V,E)$ be an instance of $L$-\lcut with source $s$ and sink $t$.
First, if $(s,t)\in E$, then that edge must be included in the solution $M$.  
Thus we can assume that edge is not in the \lcut instance, as assuming this can only make it harder to approximate the optimum value of the \lcut instance. 
We now construct an instance of \gmvid, $G'=(V',E')$. Let $V'=V$ and $E'=E\cup \{(s , t)\}$ where the edges in $E$ have weight $1$ and the edge $(s, t)$ has weight $L+1$.

First, observe that any cycle containing the edge $(s,t)$ followed by $\leq L$ unit weight edges is unbalanced, as the sum of the unit weight edges will be $<L+1=w((s,t))$.
Conversely, any unbalanced cycle must contain the edge $(s,t)$ followed by $\leq L$ unit weight edges.  
Specifically, if a cycle does not contain $(s, t)$ then it is balanced since all edges would then have weight $1$.
Moreover, if a cycle contains $(s,t)$ and $>L$ other edges, then the total sum of those unit edges will be $\geq L+1 = w((s,t))$.

Note that the edges from $G$ are exactly the weight one edges in $G'$, and thus the paths in $G$ are in one-to-one correspondence with the paths in $G'$ which consist of only weight one edges. 
Moreover, the edge $(s, t)$ in $G'$ corresponds with the source and sink from $G$. 
Thus by the above, the unbalanced cycles in $G'$ are in one-to-one correspondence with $s$-$t$-paths with length $\leq L$ in $G$.
Therefore, a minimum cardinality subset of edges which non-top cover all unbalanced cycles, i.e.\ an optimal solution to \gmvid, 
corresponds to a minimum cardinality subset of edges from $E$ which cover all paths from $s$ to $t$ of length $\leq L$, i.e.\ an optimal solution to \lcut. 
\end{proof}
}
\InFullVer{\lbcutproof}


We now reduce \gmvid to \gmvd, thus implying the above theorems also hold for \gmvd.

\begin{theorem}\thmlab{gmvid_to_gmvd}
 There is an approximation-preserving, polynomial-time reduction from \gmvid to \gmvd.
\end{theorem}
\begin{proof}
Let $G=(V,E)$ be an instance of \gmvid. 
Find the set $T=\{(s_1,t_1),\ldots, (s_{|T|},t_{|T|})\}$ of top edges of all unbalanced cycles by comparing the weight of each edge to the shortest path distance between its endpoints. 
We now construct an instance, $G'=(V',E')$, of \gmvd. 
For all $1\leq i\leq |T|$ and $1\leq j\leq |E|+1$, let $Q = \{v_{ij}\}_{i,j}$ be a vertex set, and let $F_l = \{ (s_i, v_{ij}) \}_{i,j}$ and $F_r = \{ (t_i, v_{ij}) \}_{i,j}$ be sets of edges.
Let $V' = V \cup Q$ and $E'=E \cup F_l \cup F_r$, where all $(s_i, v_{ij})$ edges in $F_l$ have weight $L=1 + \max_{e\in E} w(e)$, 
and for any $i$ all $(t_i, v_{ij})$ edges in $F_r$ have weight $L - w((s_i, t_i))$.

Let $C$ be any unbalanced cycle in $G$ with top edge $(s_i,t_i)$ for some $i$. 
First, observe that the cycle $C'=(C\setminus (s_i,t_i))\cup\{(s_i, v_{ij}), (t_i, v_{ij})\}$ is an unbalanced cycle with top edge $(s_i, v_{ij})$, for any $j$.
To see this, note that $w((s_i, v_{ij}))=L=w((t_i, v_{ij}))+w((s_i, t_i))$. Thus since $C$ is unbalanced, 
\[
w((s_i, v_{ij}))=w((t_i, v_{ij}))+w((s_i, t_i)) > w((t_i, v_{ij}))+ w(C\setminus (s_i,t_i)),
\]
and thus by definition $C'$ is unbalanced with top edge $(s_i, v_{ij})$.
Hence each unbalanced cycle $C$ in $G$, with top edge $(s_i,t_i)$, corresponds to $|E|+2$ unbalanced cycles in $G'$, 
namely, $C$ itself and the cycles obtained by replacing $(s_i, t_{i})$ with a pair $(s_i, v_{ij}),(t_i, v_{ij})$, for any $j$.

We now show the converse, that any unbalanced cycle $C'$ in $G'$ is either also an unbalanced cycle $C$ in $G$, 
or obtained from an unbalanced cycle $C$ in $G$ by replacing $(s_i, t_{i})$ with $(s_i, v_{ij}),(t_i, v_{ij})$ for some $j$. 
First, observe that for any $i$, any cycle containing the edge $(s_i,v_{ij})$ must also contain the edge $(t_i,v_{ij})$, 
and moreover, if a cycle containing such a pair is unbalanced, then its top edge must be $(s_i,v_{ij})$ as $w((s_i,v_{ij}))=L$.
Similarly, any cycle containing more than one of these pairs of edges (over all $i$ and $j$) must be balanced, since such cycles then would contain at least two edges with the maximum edge weight $L$.
So let $C'$ be any unbalanced cycle containing exactly one such $(s_i,v_{ij})$, $(t_i,v_{ij})$ pair. 
Note that $C'$ cannot be the cycle $((s_i,v_{ij}),(t_i,v_{ij}),(s_i,t_i))$, as this cycle is balanced because $w((s_i, v_{ij}))=w((t_i, v_{ij}))+w((s_i, t_i))$.
Otherwise, $C=C'\setminus\{(s_i,v_{ij}),(t_i,v_{ij})\} \cup \{(s_i,t_i)\}$ is also a cycle, and note that $C'$ being unbalanced implies $C$ is unbalanced with top edge $(s_i,t_i)$, 
implying the claim. This holds since 
\[
w(s_i,t_i) = w((s_i, v_{ij})) - w((t_i, v_{ij})) > w(C'\setminus (s_i, v_{ij})) - w((t_i, v_{ij})) = 
w(C\setminus (s_i,t_i)).
\]

Now consider any optimal solution $M$ to the \gmvid instance $G$, which by \thmref{gmvid_iff} we know is a minimum cardinality non-top cover of $G$. 
By the above, we know that $M$ is also a non-top cover of $G'$, and hence is also a regular cover of $G'$.
Thus by \thmref{gmvd_iff}, $M$ is a valid solution to the \gmvd instance.  
Conversely, consider any optimal solution $M'$ to the \gmvd instance $G'$, which by \thmref{gmvd_iff} is a minimum cardinality regular cover of $G'$.
The claim is that $M'$ is also a non-top cover of $G$, and hence is a valid solution to the \gmvid instance. 
To see this, observe that since all unbalanced cycles in $G$ are unbalanced cycles in $G'$, $M'$ must be a regular cover of all unbalanced cycles in $G$, 
and we now argue that it is in fact a non-top cover. 
Specifically, consider all the unbalanced cycles in $G$ which have a common top edge $(s_i,t_i)$.  
Suppose there is some cycle in this set, call it $C$, which is not non-top covered by $M'$. 
As $M'$ is a regular cover for $G'$, this implies that for any $j$, the unbalanced cycle described above
determined by removing the edge $(s_i, t_i)$ from $C$ and adding edges $(s_i,v_{ij})$ and $(t_i,v_{ij})$, 
must be covered either with $(s_i,v_{ij})$ or $(t_i,v_{ij})$.  
However, as $j$ ranges over $|E|+1$ values, and these edge pairs have distinct edges for different values of $j$, $M'$ has at least $|E|+1$ edges.
This is a clear contradiction with $M'$ being a minimum sized cover, as any non-top cover of $G$ is a regular cover of $G'$, and $G$ only has $|E|$ edges in total.
\end{proof}

Observe that in the above reduction, if the maximum edge weight in the \gmvid instance was $L$, then maximum edge weight in the corresponding \gmvd instance is $L+1$. 
Moreover, for the reduction in \thmref{lbcut_to_gmvid}%
\InSODAVer{ (shown in the appendix)},
for any instance of $L$-\lcut, the corresponding \gmvid instance has maximum edge weight $L+1$.
Thus based on the above reductions, and previous known hardness results, we have the following. 

\begin{theorem}
\gmvid and \gmvd are APX-complete, and moreover assuming UGC neither can be approximated within any constant factor. 

For any fixed value $L$, consider the problem defined by either the restriction of \gmvid or \gmvd to the subset of instances with maximum edge weight $L$ (and minimum edge weight 1), 
then assuming UGC this problem is hard to approximate within a factor of $\Omega(\sqrt{L})$.
\end{theorem}

Note our reduction in \thmref{lbcut_to_gmvid} actually implies a stronger second statement than what is given in the above theorem, as it reduces to an instance of \gmvid with all unit weight edges except for a single $L+1$ weight edge.
However, we find the above statement more natural, and it is sufficient for our purposes.


\section{Approximation Algorithm}\seclab{approx}

In this section we present approximation algorithms for the \gmvid and \gmvd problems.  As the algorithms are nearly identical for the two cases, we present the algorithm for \gmvd first, and then remark on the minor change needed to apply it to \gmvid.  

By \thmref{gmvd_iff}, we know that an optimal solution to \gmvd is a minimum cardinality regular cover of all unbalanced cycles.  
This naturally defines a hitting set instance $(E,\mathcal{C})$, where the ground set $E$ is the edges from $G$, and $\mathcal{C}$ is the collection of the subsets of edges determined by the unbalanced cycles.  
Thus if for any edge $e\in E$ we could compute the number of unbalanced cycles it participates in, then immediately we get an $O(\log n)$ approximation for \gmvd by running the standard greedy algorithm for hitting set. 
Namely, while unbalanced cycles remain, we simply repeatedly remove the edge hitting the largest number of remaining unbalanced cycles (as removing the edge effectively removes the sets it hit from the hitting set instance).
However, there may be an exponential number of unbalanced cycles. 
Note that in general just counting the number of simple paths in a graph is \#P-Hard \cite{v-cerp-79}, 
though it is known how to count paths of length up to roughly $O(\log n)$ using the color-coding technique. 
(See for example \cite{ag-bfphfa-10} and references therein.  Also see \cite{bdh-ec-18} for recent FPT algorithms.)
Moreover, observe that our situation is more convoluted as we only wish to count paths corresponding to unbalanced cycles.

%

Despite the challenge of counting unbalanced cycles, we are able to get an approximation by making the key observation that a cycle with the largest deficit value must correspond to a shortest path, which in turn allows us to quickly get a count when restricting to such cycles.  Thus our approach is to iteratively handle covering cycles by decreasing deficit value, ultimately breaking the problem into multiple hitting set instances.  

For any pair of vertices $s,t\in V$, we write $d(s,t)$ to denote their shortest path distance in $G$, 
and $\csp(s,t)$ to denote the number of shortest paths from $s$ to $t$.
The proof of the following is straightforward, 
\InSODAVer{and is included in the appendix for the sake of completeness.}%
\InFullVer{and is included for the sake of completeness.}

\begin{lemma}\lemlab{pathCount}
Let $G$ be a positively weighted graph, where for all pairs of vertices $u,v$ one has constant time access to the value $d(u,v)$.
Then for any pair of vertices $s,t$, the value $\csp(s,t)$ can be computed in $O(m+n\log n)$ time.  
\end{lemma}
\newcommand{\pathcountproof}{
\begin{proof}
Let $V=\{v_1,v_2,v_3,...,v_n\}$, and let $N(v_i)$ denote the set of neighbors of $v_i$. 
Define  $X_i = \{ v_j\in N(v_i)\mid w(v_i,v_j)+d(v_j,t)=d(v_i,t)\}$, that is, $X_i$ is the set of neighbors of $v_i$ where there is a shortest path from $t$ to $v_i$ passing through that neighbor.
Thus we have,
\[
\csp(s,v_i)=\sum_{v_j\in X} \csp(s,v_j).
\]
Note that any shortest path from $v_i$ to $t$ can only use vertices $v_j$ which are closer to $t$ than $v_i$. 
Thus assume we have sorted and labeled the vertices $t=v_1,v_2,v_3,...,v_n$ according to increasing order of their distance $d(v_i,t)$ from $t$.
Thus if we compute the $\csp(v_i,t)$ values in increasing order of the index $i$, then each $\csp(s,v_j)$ value can be computed in time proportional to the degree of $v_i$, 
and so the overall running time is $O(m+n\log n)$.
\end{proof}
}
\InFullVer{\pathcountproof}


Recall that for an unbalanced cycle $C$ with top edge $t$, the deficit of $C$ is $\delta(C) = w(t) - \sum_{e\in (C\setminus t)} w(e)$. 
Moreover, $\delta(G)$ is used to denote the maximum deficit over all cycles in $G$.
For any edge $e$, define $N_T(e,\alpha)$ to be the number of distinct unbalanced cycles of deficit $\alpha$ whose top edge is $e$. 
Similarly, let $N_U(e,\alpha)$ denote the number of distinct unbalanced cycles with deficit $\alpha$ which contain the edge $e$, but where $e$ is not the top edge.

\begin{lemma}
\lemlab{Topedge}
For any edge $e=(s,t)$, if $w(e)=d(s,t)+\delta(G)$ then $N_T(e,\delta(G))=\csp(s,t)$, and otherwise $N_T(e,\delta(G))=0$.
\end{lemma}
\begin{proof}
 If $w(e)\neq d(s,t)+\delta(G)$, then as $\delta(G)$ is the maximum deficit over all cycles, it must be that $w(e)< d(s,t)+\delta(G)$, 
 which in turn implies any unbalance cycle with top edge $e$ has deficit strictly less than $\delta(G)$.
 Now suppose $w(e)= d(s,t)+\delta(G)$, and consider any path $p_{s,t}$ from $s$ to $t$ such that $e$ together with $p_{s,t}$ creates an unbalanced cycle with top edge $e$. 
 If $p_{s,t}$ is a shortest path then $w(e)-w(p_{s,t}) = w(e)-d(s,t)=\delta(G)$, and otherwise $w(p_{s,t})>d(s,t)$ and so $w(e)-w(p_{s,t})<w(e)-d(s,t)=\delta(G)$.
 Thus $N_T(e,\delta(G))=\csp(s,t)$ as claimed.
\end{proof}


As $G$ is undirected, every edge $e\in E$ correspond to some unordered pair $\{a,b\}$.  
However, often we write $e=(a,b)$ as an ordered pair, 
according to some fixed arbitrary total ordering of all the vertices. 
We point this out to clarify the following statement.
\begin{lemma}
\lemlab{Nontopedge}
Fix any edge $e=(s,t)$, and let 
$X=\{f=(\alpha,\beta) \mid w(f)=d(\alpha,s)+w(e)+d(t,\beta)+\delta(G) \}$, and 
$Y=\{f=(\alpha,\beta) \mid w(f)=d(\beta,s)+w(e)+d(t,\alpha)+\delta(G) \}$.
Then it holds that
\[
N_U(e,\delta(G))=\left(\sum_{(\alpha,\beta)\in X} \csp(\alpha,s)\cdot \csp(t,\beta)\right)+\left(\sum_{(\alpha,\beta)\in Y} \csp(\beta,s)\cdot \csp(t,\alpha)\right).
\]
\end{lemma}
\begin{proof}
Consider any unbalanced cycle $C$ containing $e=(s,t)$, with top edge $f=(\alpha,\beta)$ and where $\delta(C)=\delta(G)$.
Such a cycle must contain a shortest path between $\alpha$ and $\beta$, as otherwise it would imply $\delta(G)>\delta(C)$.
Now if we order the vertices cyclically, then the subset of $C$'s vertices $\{\alpha,\beta,s,t\}$, must appear either in the order $\alpha,s,t,\beta$ or $\beta,s,t,\alpha$.
In the former case, as the cycle must use shortest paths, $w(f)=d(\alpha,s)+w(e)+d(t,\beta)+\delta(G)$, and the number of cycles satisfying this is $\csp(\alpha,s)\cdot \csp(t,\beta)$.  
In the latter case, $w(f)=d(\beta,s)+w(e)+d(t,\alpha)+\delta(G)$, and the number of cycles satisfying this is $\csp(\beta,s)\cdot \csp(t,\alpha)$.
Note also that the set $X$ from the lemma statement is the set of all $f=(\alpha,\beta)$ satisfying the equation in the former direction, and $Y$ is the set of all $f=(\alpha,\beta)$ satisfying the equation in the later direction.
Thus summing over each relevant top edge in $X$ and $Y$, of the number of unbalanced cycles of deficit $\delta(G)$ which involve that top edge and $e$, yields the total number of unbalanced cycles with deficit $\delta(G)$ 
containing $e$ as a non-top edge.
\end{proof}

\begin{corollary}\InSODAVer{[Proof in appendix]}%
\corlab{time}
 Given constant time access to $d(u,v)$ and $\csp(u,v)$ for any pair of vertices $u$ and $v$, $N_T(e,\delta(G))$ can be computed in $O(1)$ time and $N_U(e,\delta(G))$ in $O(m)$ time.
\end{corollary}
\newcommand{\cortime}{
\begin{proof}
By \lemref{Topedge}, in constant time we can check whether $w(e)=d(s,t)+\delta(G)$, in which case set $N_T(e,\delta(G))=\csp(s,t)$, and otherwise set $N_T(e,\delta(G))=0$. 
By \lemref{Nontopedge}, we can compute $N_U(e,\delta(G))$ with a linear scan of the edges, 
where for each edge $f$ in constant time we can compute whether $w(f)=d(\alpha,s)+w(e)+d(t,\beta)+\delta(G)$ and if so add $\csp(\alpha,s)\cdot \csp(t,\beta)$ to the sum over $X$, 
and similarly if $w(f)=d(\beta,s)+w(e)+d(t,\alpha)+\delta(G)$ add $\csp(\beta,s)\cdot \csp(t,\alpha)$ to the sum over $Y$.
\end{proof}
}
\InFullVer{\cortime}


\begin{algorithm2e}
    \SetKwInOut{Input}{Input}
    \SetKwInOut{Output}{Output} 
    \Input{An instance $G=(V,E)$ of \gmvd}
    \Output{A valid solution $S$ to the given instance.}
    \DontPrintSemicolon
    Let $S=\emptyset$\;
   \While{True}{
    For every pair $s,t\in V$ compute $d(s,t)$\;\label{allPairs}
    Compute $\delta(G) = \max_{e=(s,t)\in E} ~w(e)-d(s,t)$\;\label{deficits}
    \If{$\delta(G) = 0$}{
       \Return $S$\;
    }
    For every edge $(s,t)\in E$ compute $\csp(s,t)$\;\label{pathCounting}
    For every $e\in E$ compute $count(e) = N_T(e,\delta(G))+N_U(e,\delta(G))$\;\label{count} 
    Set $f = \arg\max_{e\in E} count(e)$\;
    Update $S=S\cup\{f\}$ and $G=G\setminus f$\;
   }
 \caption{Finds a valid solution for \gmd.}
\algolab{simp}
\end{algorithm2e}


\InSODAVer{
Note that the proof of the following theorem and subsequent observations are straightforward, given the above discussion, 
and thus if desired can be skipped to produce a 10-page version.
}

\begin{theorem}
 For any positive integer $c$, consider the set of \gmvd instances where the number of distict deficit values is at most $c$, i.e.\ $|\{\delta(C)\mid \text{$C$ is a cycle in $G$} \}|\leq c$.
 Then \algoref{simp} gives an $O((n^3+m^2)\cdot OPT\cdot c \log n)$ time $O(c \log n)$-approximation, where OPT is the size of the optimal solution.
\end{theorem}
\begin{proof}
 Observe that the algorithm terminates only when $\delta(G) = 0$, i.e.\ only once there are no unbalanced cycles left. 
 As no new edges are added, and weights are never modified, this implies that when the algorithm terminates it outputs a valid regular cover $S$.
 (The algorithm must terminate as every round removes an edge.)  
 Therefore, by \thmref{gmvd_iff}, $S$ is a valid \gmvd solution, and so we only need to bound its size.

 Let the edges in $S=\{s_1,\ldots,s_k\}$ be indexed in increasing order of the loop iteration in which they were selected.
 Let $G_1,\ldots, G_{k+1}$ be the corresponding sequence of graphs produced by the algorithm, where $G_i=G\setminus \{s_1,\ldots,s_{i-1}\}$.
 Note that for all $i$, $G_i=(V,E_i)$ 
 induces a corresponding instance of hitting set, $(E_i,\mathcal{C}_i)$, 
 where the ground set is the set of edges from the \gmvd instance $G_i$, and $\mathcal{C}_i = \{E_i(C)\mid \text{$C$ is an unbalanced cycle in $G_i$}\}$ (where $E_i(C)$ is the set of edges in $C$). 
 
 Let $D=\{\delta(C)\mid \text{$C$ is a cycle in $G$} \}$, where by assumption $|D|\leq c$.
 Note that any cycle $C$ in any graph $G_i$, is also a cycle in $G$. 
 Thus as we never modify edge weights, $\delta(G_1),\ldots,\delta(G_{k+1})$ is a non-increasing sequence. Moreover $X=\{\delta(G_i)\}_i\subseteq D$, and in particular $|X|\leq c$.
 For a given value $\delta\in X$, let $G_{\alpha},G_{\alpha+1},\ldots,G_{\beta}$ be the subsequence of graphs with deficit $\delta$ (this subsequence is consecutive as the deficit values are non-increasing).
 Observe that for all $\alpha\leq i\leq \beta$, the edge $s_i$ is an edge from a cycle with deficit $\delta=\delta(G_i)$.  
 So for each $\alpha\leq i\leq \beta$, define a sub-instance of hitting set $(E_i', \mathcal{C}_i')$, where $E_i'$ is the set of edges in cycles of deficit $\delta$ from $G_i$, 
 and $\mathcal{C}_i'$ is the family of sets of edges from each cycle of deficit $\delta$ in $G_i$.
 
 The claim is that for the hitting set instance $(E_\alpha', \mathcal{C}_\alpha')$, that $\{s_\alpha,\ldots,s_\beta\}$ is an $O(\log n)$ approximation to the optimal solution. 
 To see this, observe that for any $\alpha\leq i\leq \beta$ in line \ref{count}, $count(e)$ is the number of times $e$ is contained in an unbalanced cycle with deficit $\delta =\delta(G_i)$, 
 as by definition $N_T(e,\delta(G_i))$ and $N_U(e,\delta(G_i))$ count the occurrences of $e$ in such cycles as a top edge or non-top edge, respectively.
 Thus $s_i$ is the edge in $E_i'$ which hits the largest number of sets in $\mathcal{C}_i'$, 
 and moreover, $(E_{i+1}', \mathcal{C}_{i+1}')$ is the corresponding hitting set instance induced by removing $s_i$ and the sets it hit from $(E_i', \mathcal{C}_i')$. 
 Thus $\{s_\alpha,\ldots,s_\beta\}$ is the resulting output of running the standard greedy hitting set algorithm on $(E_\alpha', \mathcal{C}_\alpha')$ 
 (that repeatedly removes the element hitting the largest number of sets), and it is well known this greedy algorithm produces an $O(\log n)$ approximation.

 The bound on the size of $S$ now easily follows.  Specifically, let $I=\{i_1, i_2,\ldots,i_{|X|}\}$ be the collection of indices, where $i_j$ was the first graph considered with deficit $\delta(G_{i_j})$.  
 By the above, $S$ is the union of the $O(\log n)$-approximations to the sequence of hitting set instance $(E_{i_1}', \mathcal{C}_{i_1}'),\ldots, (E_{i_{|X|}}', \mathcal{C}_{i_{|X|}}')$. 
 In particular, note that for all $i_j$, $(E_{i_j}', \mathcal{C}_{i_j}')$ is a hitting set instance induced from the removal of a subset of edges from the initial hitting set instance $(E_1, \mathcal{C}_1)$, 
 and then further restricted to sets from cycles with a given deficit value.
 Thus the size of the optimal solution on each of these instances can only be smaller than on $(E_1, \mathcal{C}_1)$. 
 This implies that the total size of the returned set $S$ is $O(OPT\cdot |X|\log n)=O(OPT\cdot c\log n)$.
 
 As for the running time, first observe that by the above, there are $O(OPT \cdot c \log n)$ while loop iterations. For a given loop iteration, computing all pairwise distance in line \ref{allPairs} takes $O(n^3)$ time using the standard Floyd-Warshall algorithm. Computing the graph deficit in line \ref{deficits} can then be done in $O(m)$ time. For any given vertex pair $s,t$, computing $\csp(s,t)$ takes $O(m+n\log n)$ time by \lemref{pathCount}. Thus computing the number of shortest paths over all edges in line \ref{pathCounting} takes $O(m^2+mn\log n)$ time.  
 For each edge $e$, by \corref{time}, $count(e) = N_T(e,\delta(G)) + N_U(e,\delta(G))$ can be computed in $O(m)$ time, and thus computing all counts in line \ref{count} takes $O(m^2)$ time.
 As the remaining steps can be computed in linear time, each while loop iteration in total takes $O(n^3+mn\log n + m^2) = O(n^3+m^2)$ time, thus implying the running time bound over all iterations in the theorem statement. 
\end{proof}

\begin{remark}
 In the above our goal was to present the algorithm and analysis in simple and practical terms. 
 However, it should be noted that the running time can be improved, though potentially at the cost of added complication.  
 In particular, rather than computing the $d(u,v)$ values from scratch in each iteration, we can use a dynamic data structure.
 This would slightly improve the above running time to $O(n^3+(n^{2+\alpha}+m^2)\cdot OPT\cdot c \log n)$, where $0\leq \alpha$ 
 is a constant depending on the query and update time of the dynamic data structure.
 (Ignoring $\log$ factors, $\alpha=3/4$ is known. See for example the recent paper \cite{iss-fdapsp-17} and references therein).
 However, similarly improving the $m^2$ term in the running time seems more challenging as the $N_U(e,\delta(G))$ values depend in a non-trivial 
 way on collections of $d(u,v)$ values, each of which may or may not have changed.
\end{remark}

Throughout this section we considered the \gmvd problem. A similar result holds for \gmvid, and there are a couple ways to achieve it.  
First, recall that \thmref{gmvid_to_gmvd} gave a polynomial-time approximation-preserving reduction from \gmvid to \gmvd, 
and hence this reduction and the above algorithm could be combined to yield an algorithm for \gmvid.  
One issue with this approach is that the reduction in \thmref{gmvid_to_gmvd} increases the graph size by a linear factor, resulting in a slower running time.  
An alternative and simpler approach is to observe that the above proof and algorithm will work nearly identically for \gmvid,
except that by \thmref{gmvid_iff} the sets in the corresponding hitting set instance should not include the top edge of each cycle. 
We thus have the following.


\begin{theorem} 
 For any positive integer $c$, consider the set of \gmvid instances where the number of distict deficit values is at most $c$, i.e.\ $|\{\delta(C)\mid \text{$C$ is a cycle in $G$} \}|\leq c$.
 Then \algoref{simp}, where line \ref{count} instead sets $count(e) = N_U(e,\delta(G))$, gives an $O((n^3+m^2)\cdot OPT\cdot c \log n)$ time $O(c \log n)$-approximation.
\end{theorem}

Note that the \gmvd and \gmvid problems are phrased in terms of finding minimum sized edge sets whose weights can be modified. 
To determine how to modify the weights of the output edges, recall that this can be done with the LP in \secref{prelims}.


\bibliographystyle{alpha}
\bibliography{sample}

\InSODAVer{
\appendix
\section{Missing proofs}


\subsection*{Proof from \secref{prelims}}
Consider the linear program described in \secref{prelims}.
\lpproof


\subsection*{Proofs from \secref{cover}}




\helper

\noindent
\textbf{\thmref{gmvd_iff}.}%
\textit{
If $G$ is an instance of \gmvd and $S$ is a regular cover of all unbalanced cycles, then $G$ can be converted into a metric graph by only changing weights of edges in $S$.
}
\gmvdproof


\subsection*{Proof from \secref{hard}}

\noindent
\textbf{\thmref{lbcut_to_gmvid}.}%
\textit{
For any fixed value $L$, there is an approximation-preserving, polynomial-time reduction from $L$-\lcut to \gmvid.
}
\lbcutproof


\subsection*{Proofs from \secref{approx}}

\noindent
\textbf{\lemref{pathCount}.}%
\textit{
Let $G$ be a positively weighted graph, where for all pairs of vertices $u,v$ one has constant time access to the value $d(u,v)$.
Then for any pair of vertices $s,t$, the value $\csp(s,t)$ can be computed in $O(m+n\log n)$ time. 
}
\pathcountproof


\noindent
\textbf{\corref{time}.}%
\textit{
 Given constant time access to $d(u,v)$ and $\csp(u,v)$ for any pair of vertices $u$ and $v$, $N_T(e,\delta(G))$ can be computed in $O(1)$ time and $N_U(e,\delta(G))$ in $O(m)$ time.
}
\cortime
}

\end{document}